\documentclass[11pt]{article}  

\usepackage{fullpage}
\usepackage{amsmath}
\usepackage{amssymb}
\usepackage{graphicx}
\usepackage{citesort}
\usepackage{multirow}
\usepackage{amsthm}

\newcommand{\eps}{\varepsilon}

\newtheorem{definition}{Definition}[section] 
\newtheorem{Definition}{Definition}[section] 
\renewcommand{\paragraph}[1]{\medskip\noindent{\bf #1}}
 
\newcommand{\abs}[1]{\left| #1 \right|}
\newcommand{\disc}{\mathrm{disc}}
\newcommand{\Acal}{\mathcal{A}}

\newcommand{\Rcal}{\mathcal{R}}

\newcommand{\Zcal}{\mathcal{Z}}
\newcommand{\Tbf}{\mathbf{T}}
\newcommand{\Zbf}{\mathbf{Z}}
\newcommand{\Tcal}{\mathcal{T}}

\newcommand{\Pcal}{\mathcal{P}}

\newtheorem{theorem}{Theorem}[section]
\newtheorem{lemma}{Lemma}[section]
\newtheorem{corollary}{Corollary}[section]

\newtheorem{fact}{Fact}[section]

\numberwithin{equation}{section}

\begin{document}
\markboth{Z. Wei and K. Yi}{The Space Complexity of 2-Dimensional Approximate Range
    Counting and Combinatorial Discrepancy}

\title{The Space Complexity of 2-Dimensional Approximate Range
    Counting and Combinatorial Discrepancy\thanks{A
    preliminary version of the paper appeared in SODA'13.}}

\author{Zhewei Wei \thanks{School of Information, Renmin University of
    China. zhewei@ruc.edu.cn} 
Ke Yi \thanks{Hong Kong University of Science and Technology. yike@cs.ust.hk}
}
\maketitle
\begin{abstract} \small
We study the problem of $2$-dimensional orthogonal range counting with
additive error. Given a set $P$ of $n$ points drawn from an $n\times n$
grid and an error parameter $\eps$, the goal is to build a data structure,
such that for any orthogonal range $R$, it can return the
number of points in $P\cap R$ with additive error $\eps n$. A well-known
solution for this problem is the {\em $\eps$-approximation}, which
 is a subset $A\subseteq P$ that
can estimate the number of points in $P\cap R$ with the
number of points in $A\cap R$.  It is known that an $\eps$-approximation of
size $O(\frac{1}{\eps} \log^{2.5} \frac{1}{\eps})$ exists for any $P$ with
respect to orthogonal ranges, and the best lower bound is
$\Omega(\frac{1}{\eps} \log \frac{1}{\eps})$.

The $\eps$-approximation is a rather restricted data structure, as we are
not allowed to store any information other than the coordinates of the
 points in $P$.  In this paper, we explore what can be achieved without
any restriction on the data structure. We first describe a simple data structure
that uses $O(\frac{1}{\eps}(\log^2\frac{1} {\eps} +  
\log n) )$ bits and answers queries with error $\eps n$. We then prove a
lower bound that any data structure that answers queries with error
$\eps n$ must use $\Omega(\frac{1}{\eps}(\log^2\frac{1} {\eps} +  
\log n) )$ bits. Our lower bound is information-theoretic: We show that there is a
collection of $2^{\Omega(n\log n)}$ point sets with large {\em union
  combinatorial discrepancy}, and thus are hard to distinguish unless
we use $\Omega(n\log n)$  bits. 
\end{abstract}

\section{Introduction}
Range counting is one of the most fundamental problems in computational
geometry and data structures. Given $n$ points in $d$ dimensions, the goal
is to preprocess the points into a data structure, such that the number of
points in any query range can be returned.  Range counting has been studied
intensively, and a lot of work has focused on the space-query time tradeoff
or the update-query tradeoff of the data structure. We refer the reader to the
survey by Agarwal and Erickson~\cite{ae-grsr-97} for these results.  In this paper, we
look at the problem from a data summarization/compression point of view:
What is the minimum amount of space that is needed to encode all the range
counts approximately?  Approximation is necessary here, since otherwise we
will have to remember the entire the point set.  It is also easy to see
that relative approximation will not help either, as it requires us to
differentiate between empty ranges and those containing only one point.
Thus, we aim at an absolute error guarantee.  As we will be dealing with
bit-level space complexity, it is convenient to focus on an integer grid.
More formally, we are given  a set of $n$ points $P$ drawn from an
$n \times n$ grid and an error parameter $\eps$ . The goal is to build a
data structure, such that for any orthogonal range $R$, the data structure
can return the number of points in $P\cap R$ with additive error $\eps n$.
 
We should mention that there is another notion of approximate range
counting that approximates the range, i.e., points near the boundary of the
range may or may not be counted \cite{arya06}.  Such an approximation
notion clearly precludes any sublinear-space data structure as well.

\subsection{Background and related results}
$\;$

\paragraph{$\eps$-approximations.} Summarizing point sets while preserving
range counts (approximately) is a fundamental problem with applications in
numerical integration, statistics, and data mining, among many others.  The
classical solution is to use the {\em $\eps$-approximation} from
discrepancy theory.  Consider a range space $(P,\Rcal)$, where $P$ is a
finite point set of size $n$. A subset $A \subseteq P$ is called an
$\eps$-approximation of $(P, \Rcal)$ if
$$\max_{R\in \Rcal}\abs{\frac{\abs{R\cap A}}{\abs{A}}-\frac{\abs{R\cap
      P}}{\abs{P}}}\le \eps. $$
This means that we can approximate $|R\cap
P|$ by counting the number of points in $R\cap A$ and scaling back, with
error at most $\eps n$.

Finding $\eps$-approximations of small size for various geometric range spaces has been a
central research topic in computational geometry. Please see the books by Matousek
\cite{matousek:discrepency} and Chazelle \cite{chazelle00:_discr} for a comprehensive
coverage on this topic.  Here we only review the most relevant results, i.e., when the
range space is the set of all orthogonal rectangles in $2$ dimensions, which we denote as
$\Rcal_2$.  This question dates back to Beck \cite{beck1981balanced}, who showed that
there are $\eps$-approximations of size $O(\frac{1}{\eps}\log^4\frac{1}{\eps})$ for any
point set $P$.  This was later improved to
$O\left(\frac{1}{\eps}\log^{2.5}\frac{1}{\eps}\right)$ by Srinivasan
\cite{srinivasan1997improving}.  These were not constructive due to the use of a
non-constructive coloring with combinatorial discrepancy $O(\log^{2.5} n)$ for orthogonal
rectangles.  Recently, Bansal~\cite{bansal2010constructive} and
Lovett et al.~\cite{lovett2012constructive} proposed algorithms to
construct such a coloring, and therefore has made these results constructive.  On the
lower bound side, it is known that there are point sets that require $\eps$-approximations
of size $\Omega(\frac{1}{\eps}\log\frac{1}{\eps})$ \cite{beck1981balanced}.

\paragraph{Combinatorial discrepancy.} 
Given a range space $(P,\Rcal)$ and a coloring function $\chi:P\rightarrow
\{-1,+1\}$, we define the {\em discrepancy} of a range $R \in \Rcal$ under $\chi$ to
be
 $$\chi(P\cap R)=\sum_{p\in P\cap R} \chi(p).$$
The discrepancy of the range space $(P,\Rcal)$ is defined as
$$\disc(P,\Rcal)=\min_{\chi}\max_{R\in \Rcal} \abs{\chi(P\cap R)},$$
namely, we are looking at the coloring that minimizes the color difference
of any range in $\Rcal$.  This kind of discrepancy is called {\em
  combinatorial discrepancy} or sometimes {\em red-blue discrepancy}. Taking the maximum over all point sets of size
$n$, we say that the combinatorial discrepancy of $\Rcal$ is
$\disc(n,\Rcal)=\max_{\abs{P}=n}\disc(P,\Rcal)$.

There is a close relationship between combinatorial discrepancy and $\eps$-approximations,
as observed by Beck \cite{beck1981balanced}.  For orthogonal ranges, the relationship is
particularly simple: The combinatorial discrepancy is at most $t(n)$ if and only if there
is an $\eps$-approximation of size $O(\frac{1}{\eps}t(\frac{1}{\eps}))$.  In fact, all the
aforementioned results on $\eps$-approximations follow from the corresponding results on
combinatorial discrepancy.  So the current upper bound on the combinatorial discrepancy of
$\Rcal_2$ is $O(\log^{2.5} n)$ \cite{srinivasan1997improving}.  The lower bound is
$\Omega(\log n)$ \cite{beck1981balanced}, which follows from the Lebesgue discrepancy
lower bound (see below).  Closing the $\Theta(\log^{1.5}n)$ gap between the upper and the
lower bound remains a major open problem in discrepancy theory.  For orthogonal ranges in
$d\ge 3$ dimensions, the current best upper bound is $O(\log^{d+1/2} n)$ by Larsen
\cite{larsen11:_range}, while the lower bound is $\Omega((\log
n)^{d-1})$ , which is recently proved by Matou{\v{s}}ek and Nikolov~\cite{matouvsek2015combinatorial}.

\paragraph{Lebesgue discrepancy.}
Suppose the points of $P$ are in the unit square $[0,1)^2$.  The Lebesgue
discrepancy of $(P, \Rcal)$ is defined to be
$$D(P,\Rcal)=\sup_{R\in \Rcal} \abs{\abs{P\cap R}-\abs{R\cap [0,1)^2}}.$$
The Lebesgue discrepancy describes how uniformly the point set $P$ is
distributed in $[0,1)^2$. Taking the infimum over all point sets of size
$n$, we say that the Lebesgue discrepancy of $\Rcal$ is 
$D(n,\Rcal)=\inf_{\abs{P}=n}D(P,\Rcal)$. 

The Lebesgue discrepancy for $\Rcal_2$ is known to be $\Theta(\log n)$.  The lower bound
is due to Schmidt~\cite{schmidt1972irregularities}, while there are many point sets (e.g.,
the Van der Corput sets~\cite{van1936verteilungsfunktionen} and the $b$-ary
nets~\cite{sobol1967distribution}) that are proved to have $O(\log n)$ Lebesgue
discrepancy.  It is well known that the combinatorial discrepancy of a range space cannot
be lower than its Lebesgue discrepancy, so this also gives the $\Omega(\log n)$ lower
bound on the combinatorial discrepancy of $\Rcal_2$ mentioned above.

\paragraph{$\eps$-nets.}
For a range space $(P,\Rcal)$, a subset $A\subseteq P$ is called an
$\eps$-net of $P$ if for any range $R\in \Rcal$ that satisfies $\abs{P\cap
  R}\ge \eps n$, there is at least $1$ point in $A\cap R$. Note that an
$\eps$-approximation is an $\eps$-net, but the converse may not be
true. 

For a range space $(P,\Acal)$, Haussler and Welzl~\cite{hw-ensrq-87} show that if the
range space has finite VC-dimension $d$, there exists an $\eps$-net of size
$O(\frac{d}{\eps} \log \frac{d}{\eps})$. For $\Rcal_2$, the current best construction is
due to Aronv, Ezra and Sharir~\cite{aronov2010small}, which has size $O(\frac{1}{\eps}
\log \log \frac{1}{\eps})$. A recent result by Pach and Tardos~\cite{pach2013tight} shows
that this bound is essentially optimal.  For more results on $\eps$-nets, please refer to
the book by Matou{\v{s}}ek~\cite{matousek:discrepency}.  In this paper, our data structure
will be based an {\em $\eps$-net} for $\Rcal_2$.

\paragraph{Approximate range counting data structures.}
The $\eps$-approximation is a rather restricted data structure, as we are
not allowed to store any information other than the coordinates of a subset
of points in $P$. In this paper, we explore what can be achieved without
any restriction on the data structure.  In 1 dimension, there is nothing
better: An $\eps$-approximation has size $O(\frac{1}{\eps})$, which takes
$O(\frac{1}{\eps} \log n)$ bits.  On the other hand, simply consider the
case where the $n$ points are divided into groups of size $\eps n$, where
all points in each group have the same location.  There are $n^{1/\eps}$
such point sets and the data structure has to differentiate all of them.
Thus $\log(n^{1/\eps}) = \frac{1}{\eps} \log n$ is a lower bound on the
number of bits used by the data structure.

Finally, we remark that there are also other work on approximate range
counting with various error measure, such as relative
$\eps$-approximation~\cite{har2011relative}, relative error data
structure~\cite{afshani2009approximate,aronov2010approximate}, and
absolute error model~\cite{arya06}. These error measures are different
from ours, and it is not clear if these problems admit sublinear space solutions.

\subsection{Our results} 
This paper settle the following problem: How many bits do we need to
encode all the orthogonal range counts with
additive error $\eps n$ for a point set on the plane?  
We first show that if we are allowed to store extra information
other than the coordinates of the points, then there is a data structure that uses
$O(\frac{1}{\eps}(\log^2\frac{1} {\eps} +  \log n) )$ bits.  
This is a $\Theta(\log^{1.5}\frac{1}{\eps})$ improvement from
$\eps$-approximations.  

The majority of the paper is the proof of a matching lower bound:  
We show that for $\eps \ge c\log n /n$ for some constant $c$,
any data structure that answers queries with error $\eps n$ must use
$\Omega(\frac{1}{\eps}(\log^2\frac{1} {\eps} +  \log n) )$  bits. In
particular, if we set $\eps = c\log n /n$, then any data structure that answers queries with error $\eps n$
must use $\Omega(n \log n)$ bits, which implies that that answering
queries with error $O(\log n)$ is as hard as answering the queries exactly.

The core of our lower bound proof  is the construction of a
collection $\Pcal^*$ of $2^{\Omega(n\log n)}$ point sets with large
{\em union combinatorial discrepancy}.  More precisely, we show
that the union of any two point sets in $\Pcal^*$ has high combinatorial
discrepancy, i.e., at least $c\log n$.  Then, for any two point sets $P_1,
P_2 \in \Pcal^*$, if $\disc(P_1 \cup P_2, \Rcal_2)\ge c\log n$, that means
for any coloring $\chi$ on $P_1\cup P_2$, there must exist a rectangle $R$
such that $|\chi(R)| \ge c\log n$.  Consider the coloring $\chi$ where
$\chi(p)=1$ if $p\in P_1$ and $\chi(p)=-1$ if $p \in P_2$.  Then there
exists a rectangle $R$ such that $|\chi(R)| = \abs{\abs{R\cap P_1}-\abs{R
    \cap P_2}}\ge c \log n$.  This implies that a data structure that
answers queries with error $\frac{c}{2}\log n$ have to distinguish $P_1$
and $P_2$.  Thus, to distinguish all the $2^{\Omega(n\log n)}$ point sets in
$\Pcal^*$, the data structure has to use at least $\Omega(n\log n)$
bits, which is a tight lower bound for $\eps = n/\log n$. We will show
how the combinatorial discrepancy bound implies tight lower bound for
arbitrary $\eps$ in Section~\ref{sec:lower_bound}.

While point sets with low Lebesgue discrepancy or high combinatorial
discrepancy have been extensively studied, here we have constructed a large collection
of point sets in which the pairwise union has high combinatorial
discrepancy.  This particular aspect appears to be novel, and our
construction could be useful in proving other space lower bounds.
It may also have applications in situations where we need a ``diverse''
collection of (pseudo) random point sets.

\section{Upper Bound}
\label{sec:data_structure}
In this section, we build a data structure that supports approximate range
counting queries. Given a set of $n$ points on an $n \times n$ grid, our
data structure uses $O(\frac{1}{\eps}(\log^2\frac{1} {\eps} +
\log n) $
 bits and answers an orthogonal range counting query
with error $\eps n$. We note that it is sufficient to only consider
two-sided ranges, since an $4$-sided range counting query can be expressed as a linear
combination of four two-sided range counting queries by the inclusion-exclusion
principle.  A
two-sided range is specified by a rectangle of the form $[0,x)\times [0,y)$, where
$(x,y)$ is called the {\em query point}.

\paragraph{The data structure.}
Our data structure is an approximate variant of Chazelle's 
linear-space version of the range tree, originally for exact 
orthogonal range counting~\cite{chazelle1988functional}. 
Consider a set $P$ of $n$ points on an $n\times n$ grid. 
We divide $P$ into the left point set $P_L$ and the right point set
$P_R$ by the median of the $x$-coordinates. We will recursively build a data structure 
for $P_L$ and $P_R$. Let $B$ be a parameter to be determined later.
Let $Q(P)$ denote the ${n\over B}$ quantiles of the $y$-coordinates of
$P$. Note that  the $i$-th
quantile is the $y$-coordinate in $P$ with exactly $iB$
points below it.
We use indices $[{n\over B}] = 1, \ldots, {n\over B}$ to represent $Q(P)$, where $i$
denote the $i$-th quantile.
We don't explicitly store the $y$-values or even the indices of $Q(P)$. Instead, 
for each index $i$ in $Q(P)$ with coordinate $y$, we store a pointer
to the successor of $y$ in $Q(P_L)$. 
Note that these ${n\over B}$ pointers form 
a monotone increasing sequence of ${n\over B}$ indices in $[{n \over 2B}]$, and 
can be encoded in $O({n\over B})$ bits. Similarly, we store the
successor pointers from
$Q(P)$ to $Q(P_R)$ with $O({n\over B})$ bits.
It follows that the space in bits satisfies recursion $S(n) = 2S({n \over 2})+O({n\over B})$, with base case $S(B)=0$. 
The recurrence solves to $S(n) = O({n\over B}\log{n\over B})$. Finally, we
explicitly store the ${1\over \eps}$ quantiles $Q_0(P)$ for the
$y$-coordinates of $P$ with $O({1 \over \eps}
\log n)$ bits.

Given a query $q=(q.x, q.y)$. For simplicity, we assume $q.y$ is in
$Q_0(P)$. If not, we can use the successor of $q.y$ in $Q_0(P)$ as an
estimation with additive error at most $\eps n$ to
the final count.  If $q$ is in $P_L$,
we follow the pointer to find the successor of $q.y$ in $Q.L$, and the
recurse the problem in $P_L$. If $q$ is in $P_r$, we first  
follow the pointer to get the successor of $q.y$ in $Q(P_L)$. This gives 
an approximate count for $P_L \cap q$ with additive error $B$. We then
follow the pointer to get the successor of $q.y$ in $Q(P_R)$,  
and recurse the problem in $P_R$. Note that  
rounding $q.y$ with the successor in $P_R$ or $R_L$ causes additive
error $B$, and using the approximate count for $P_L\cap q$ also causes
additive error $B$.   
Thus, the overall additive error satisfies $E(n) = E({n \over 2})+2B$, with base case $E(B)=B$. 
The recurrence solves to $E(n) = O(B \log{n\over B})$, and 
we can then set $B=\eps n/\log{1 \over \eps}$ to make $E(n)=O(\eps n)$. 
It follows that $S(n) = O({1 \over \eps}\log^2{1\over \eps})$, and
thus total space usage is
$O(\frac{1}{\eps}(\log^2 \frac{1}{\eps} + 
\log n))$ bits.  
The query time can also be made $O(\log{1 \over \eps})$, if we use 
succinct rank-select structures to encode the pointers, as 
in Chazelle's method.

\begin{theorem}
\label{thm:upper_bound}
Given a set of $n$ points drawn from an $n \times n$ grid, there is a data structure
that uses $O(\frac{1}{\eps}(\log^2 \frac{1}{\eps} +
\log n))$ bits and answers orthogonal range counting query with additive error $\eps n$.
\end{theorem}

\section{Lower Bound}
\label{sec:lower_bound}
In this section, we prove a lower bound that matches the upper bound
in Theorem~\ref{thm:upper_bound}.

\begin{theorem}
\label{thm:lower_bound}
Consider a set of $n$ points drawn from an $n \times n$ grid.
A data structure that answers orthogonal range counting query with
additive error $\eps n$ for any point set must use $\Omega(\frac{1}{\eps}(\log^2 \frac{1}{\eps} +
\log n))$ bits.
\end{theorem}

To prove Theorem~\ref{thm:lower_bound}, we need the following theorem
on union discrepancy.

\begin{theorem}
\label{thm:point_sets} Let $\Pcal$ denote the collection of all $n$-point
sets drawn from an
$n\times n$ grid. There exists a constant $c$ and a sub-collection
$\Pcal^*\subseteq \Pcal $ of size $2^{\Omega(n\log n)}$, such that for
any two point sets $P_1, P_2 \in \Pcal^*$,  their union discrepancy $\disc(P_1\cup P_2, \Rcal_2) \ge c \log n$. 
\end{theorem} 

We first show how Theorem~\ref{thm:point_sets} implies Theorem~\ref{thm:lower_bound}.

\begin{proof}[of Theorem~\ref{thm:lower_bound}] 
We only need to prove the $\Omega({1
  \over \eps} \log^2 {1 \over \eps})$ lower bound. Suppose we group the points into
$N = {1 \over \eps} \log {1 \over \eps}$ 
fat points, each of size $\eps n / \log {1 \over \eps}$. By Theorem~\ref{thm:point_sets}, there is a
collection $\Pcal^*$ of $2^{\Omega(N\log N)}$ fat point sets, such that
for any two fat point sets $P_1, P_2 \in \Pcal^*$, there exists a
rectangle $R$ such that the number of fat points in $R\cap P_1$ and $R
    \cap P_2$  differs by at least $\ge c \log N$. Since each fat
    points corresponds to $\eps n / \log {1 \over \eps}$ points, it follows that the counts of $P_1\cap R$
and $P_2 \cap R$ differs by at least  
$${\eps n\over  \log {1 \over \eps}}\cdot c\log N = {\eps n\over  \log {1 \over \eps}} \cdot c\log \left ({1 \over \eps}
\log {1 \over \eps }\right) \ge c\eps
n.$$
 Therefore, a data structure that
answers queries with error ${c \over 2}\eps n$ have to distinguish $P_1$
and $P_2$.  
Thus, to distinguish all the $2^{\Omega(N\log N)}$ point sets in
$\Pcal^*$, the data structure has to use at least $\Omega(N\log N) =
\Omega({1 \over \eps} \log^2 {1 \over \eps})$ bits.
\end{proof}

In the rest of this section, we will focus on proving
Theorem~\ref{thm:point_sets}.  To derive the sub-collection $\Pcal^*$ in
Theorem~\ref{thm:point_sets}, we begin by looking into  a collection of point sets
called {\em binary nets}. Binary nets are a
special type of point sets under a more general concept called {\em
  $(t,m,s)$-nets}, which are introduced in~\cite{matousek:discrepency} as
an example of point sets with low Lebesgue discrepancy. See the survey by
Clayman et~al.~\cite{clayman1999updated} or the book by Hellekalek
et~al.~\cite{hellekalek1998random} for more results on $(t,m,s)$-nets.  In
this paper we will show that binary nets have two other nice properties: 1)
A binary net has high combinatorial discrepancy, i.e., $\Omega(\log n)$; 2)
there is a bit vector representation for every binary net, which allows us
to extract a sub-collection by constructing a subset of bit vectors.  In
the following sections, we will define binary nets, and formalize these two
properties.

\subsection{Definitions}
For ease of the presentation, we assume that the $n\times n$ grid is
embedded in the unit square $[0,1)^2$. We partition $[0, 1)^2$ into $n\times
n$  squares, each of size $\frac{1}{n^2}$. We assume the grid points are placed at the mass 
centers of the  $n^2$ squares, that is, each grid point has
coordinates $(\frac{i}{n}+\frac{1}{2n}, \frac{j}{n}+\frac{1}{2n})$, for $i,j \in [n]$, where $[n]$ denote
the set of all integers in $[0,n)$.   For the
sake of simplicity, we define the grid point $(i,j)$ to be
the grid point with coordinates $(\frac{i}{n}+\frac{1}{2n}, \frac{j}{n}+\frac{1}{2n})$, and we 
do not distinguish a grid point and the square
it resides in. 

Now we introduce the concepts of {\em $(a,b)$-cell} and {\em $k$-canonical}
cell.

\begin{Definition}
A {\em $(a,b)$-cell} at position $(i,j)$ is  the
rectangle  $[\frac{i2^a}{n},\frac{(i+1) 2^a}{n}) \times [\frac{ j2^b}{n}, \frac{(j+1) 2^b}{n})$.  We use
$G_{a,b}(i, j)$ to denote the $(a,b)$-cell at position $(i,j)$, and  $G_{a,b}$
to denote the set of all $(a,b)$-cells.
\end{Definition}

\begin{Definition}
A {\em $k$-canonical cell} at position $(i,j)$
 is a $(k, \log n-k)$-cell with coordinates $(i,j)$. We use
 $G_k(i,j)$, to denote the $k$-canonical cell at position $(i,j)$, and
 $G_k$ to denote the set of all $k$-canonical cells.
\end{Definition}
 
   
\begin{figure}[t]
\centering{
\includegraphics[width=.7\linewidth]{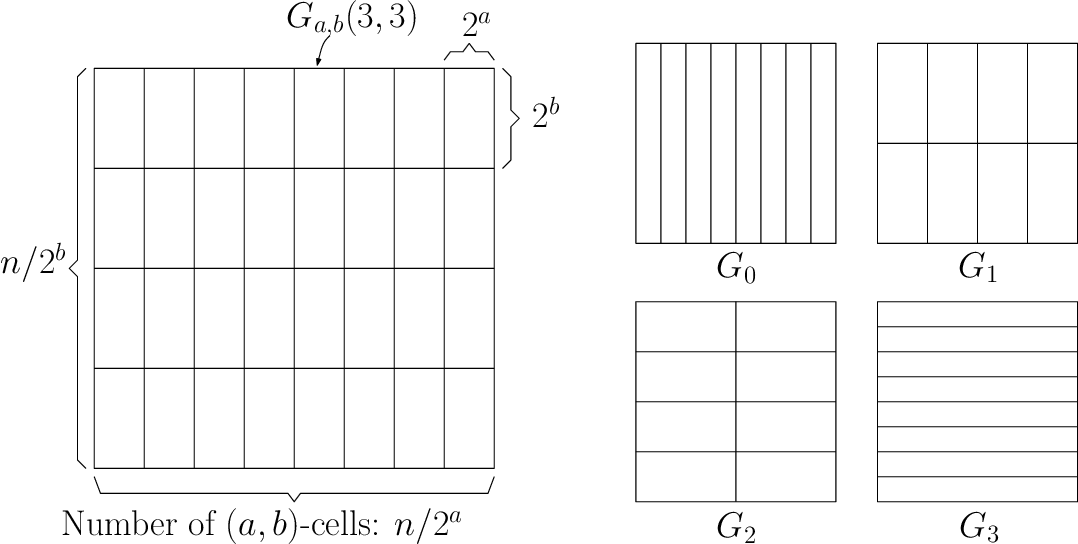}}  
\caption{\label{fig:ab} Illustrations of  $(a,b)$-cells and canonical cells.}
\end{figure}

Figure~\ref{fig:ab} is the illustration of $(a,b)$-cells and
canonical cells. Note that the position $(i,j)$ for a $(a,b)$-cell takes value in
$[n/2^a]\times [n/2^b]$. 
In particular, we call $G_0(i,0)$ the $i$-th column and
$G_{\log n}(0,j)$ the $j$-th row. Note that for a fixed $k$, $G_k$ 
partitions the grid $[0,1)^2$ into
$n$ rectangles. Based on the definition of $k$-canonical cells, we
define the binary nets:
 
\begin{Definition}
A point set $P$ is  called a  binary net if for any $k \in
[\log n]$, $P$ has exactly one point in each $k$-canonical
cell.
\end{Definition}

Let $\Pcal_0$ denote the collection of binary nets.  
In other word, $\Pcal_0$ is the set
$$\{ P \mid \abs{P\cap G_k(i,j)}=1, k \in [\log n], i \in [n/2^k],
j \in [2^k] \}.$$
  
It is known that the point sets in $\Pcal_0$ have Lebesgue discrepancy
$O(\log n)$; below we show that they also have $\Omega(\log n)$ combinatorial
discrepancy.  However, the union of two point sets in $\Pcal_0$ could have combinatorial
discrepancy as low as $O(1)$.   Thus we need to carefully extract a subset
from $\Pcal_0$ with high pairwise union discrepancy.

\subsection{Combinatorial Discrepancy and Corner Volume}
In this section, we focus on proving the following theorem, which
shows that the combinatorial discrepancy of a binary net is large.

\begin{theorem}
\label{thm:discrepancy}
For  any point set $P\in \Pcal_0$, we have $\disc(P,\Rcal_2) =\Omega(\log n)$. 
\end{theorem}

Strictly speaking, Theorem~\ref{thm:point_sets} does not depend on
Theorem~\ref{thm:discrepancy}, but this theorem gives
us some insights on the binary nets. Moreover, a key lemma
to proving Theorem~\ref{thm:point_sets}
(Lemma~\ref{lem:corner_volume_distance}) shares essentially the same proof
with Theorem~\ref{thm:discrepancy}. To prove Theorem~\ref{thm:discrepancy}, we need the following definition of {\em corner
  volume}:

\begin{Definition}
Consider a point set $P\in \Pcal_0$ and a $k$-canonical cell
$G_k(i,j)$. Let $q$ be the point of $P$ in $G_k(i,j)$. We define the {\em corner
volume} $V_P(k,i,j)$ to be the volume of the orthogonal rectangle defined by $q$
and its nearest corner of $G_k(i,j)$. We use $S_P$ to denote the
summation of the corner volumes over all possible triples $(k,i,j)$, that
is, 
$$S_P=\sum_{k=0}^{\log n}\sum_{i=0}^{n/2^k-1}\sum_{j=0}^{2^k-1}
V_P(k,i,j).$$
\end{Definition}

\begin{figure}[t]
\centering{ 
\includegraphics[width=.6\linewidth]{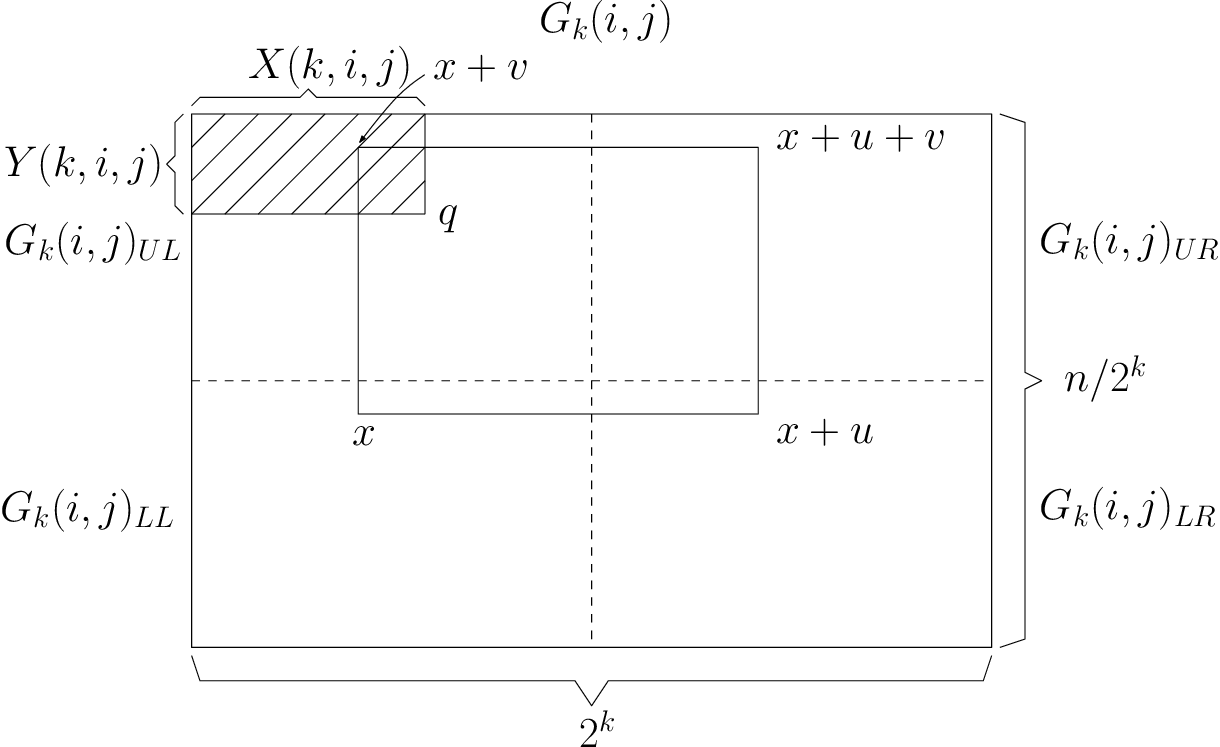}}
\caption{\label{fig:integration} Illustration of  the corner volume
  and the four analogous points. The area in shadow represents the
  corner volume $V_P(k,i,j).$}
\end{figure}

See Figure~\ref{fig:integration} for the illustration of corner
volumes. A key insight of our lower bound proof is  
the following lemma, which relates the combinatorial discrepancy of $P$ with
its corner volume sum $S_P$.
\begin{lemma}
\label{lem:corner}
There exists a constant $c$, such that 
for any point set $P \in \Pcal_0$ with corner volume sum
$$S_P \ge c\log n,$$
we have  $\disc(P,\Rcal_2)= \Omega (\log n)$.
\end{lemma} 

The proof of Lemma~\ref{lem:corner} 
makes use of the Roth's orthogonal function
method~\cite{roth1954irregularities}, which is 
widely used for proving lower bounds for Lebesgue discrepancy
(see~\cite{chazelle00:_discr,matousek:discrepency}). 
\begin{proof}
Consider a binary net $P\in \Pcal_0$ that satisfies $S_P
\ge c \log n$, where $c$ is constant to be determined later.  Given any coloring $\chi: P \rightarrow \{-1,
+1\}$ and a point  $x=(x_1, x_2) \in [0, 1)^2$, the combinatorial
discrepancy $D(x)$ at a point $x$ is defined to be 
$$D(x) = \sum_{p \in P \cap [0,x_1) \times [0, x_2)} \chi (p).$$ 
If we can  prove  $\sup_{x\in [0,n)^2}\abs{D(x)}= \Omega (\log n)$, 
the lemma will follow.

For  $k \in [\log n]$, we define normalized wavelet functions $f_k$ as
follow: for each $k$-canonical cell $G_k(i,j)$, let $q$ denote the
point contained in it. We subdivide $G_k(i, j)$ into four
equal-size quadrants, and use $G_k(i,j)_{UR}$, $G_k(i,j)_{UL}$,
$G_k(i,j)_{LR}$, $G_k(i,j)_{LL}$ to denote the upper right, upper
left, lower right and lower left quadrants, respectively (See
Figure~\ref{fig:integration}). 
 Set $f_k(x)= \chi(q)$ over quadrants  $G_k(i,j)_{UR}$ and
 $G_k(i,j)_{LL}$,   and $f_k(x)=-\chi(q)$ over the other two
quadrants. To truly reveal the power of these wavelet functions, we define a more general
class of functions called {\em checkered} functions.

\begin{definition}
\label{def:checkered} We say a function $f:[0,1)^2
\rightarrow \mathbb{R}$ is 
{\em $(a,b)$-checkered} if for each $(a,b)$-cell, there
exists a color $C\in \{-1,+1\}$ such that $f$ is equal to $C$ over
$G_{a,b}(i,j)_{UR}$ and $G_{a,b}(i,j)_{LL}$ and $-C$
over the other two quadrants.
\end{definition}

Note that our definition of checkered function is slight different
from the one used in~\cite{chazelle00:_discr}. It is easy to see the wavelet function $f_k$ is
$(k,\log n-k)$-checkered, and the integration
of a $(a,b)$-checkered function  over an $(a,b)$-cell is $0$. The
following lemma states that the checkered property is ``closed'' under multiplication. 

\begin{fact}
\label{fact:checkered}
If $f$ is $(a_1,b_1)$-checkered and $g$ is $(a_2,b_2)$ checkered, where $a_1
< a_2$ and $b_1 > b_2$, then $fg$ is $(a_1,b_2)$-checkered.
\end{fact}

For a proof, consider an $(a_1, b_2)$-cell $G_{a_1,b_2}(i,j)$. We observe that
this cell is defined by the intersection of an $(a_1,b_1)$-cell and an
$(a_2,b_2)$-cell, and we use  $G_{a,b}(i_1,j_1)$ and
$G_{a_2,b_2}(i_2,j_2)$ to denote these two cells,
respectively. Therefore the
four quadrants of $G_{a_1,b_2}(i,j)$ are defined by the
intersections of two neighboring quadrants of $G_{a_1,b_1}(i_1,j_1)$ and
two neighboring quadrants of $G_{a_2,b_2}(i_2,j_2)$. Without
loss of generality, we assume the four quadrants are defined by the intersections of the two
upper quadrants of $G_{a_1,b_1}(i_1,j_1)$ and two left quadrants of
$G_{a_2,b_2}(i_2,j_2)$ (see Figure~\ref{fig:checkered}).   Since $f$
is $(a_1,b_1)$-checkered and $g$ is $(a_2,b_2)$ checkered,  we can assume $f$
equal to $C_1$ and $-C_1$ over $G_{a_1,b_1}(i_1,j_1)_{UR}$ and
$G_{a_2,b_2}(i_2, j_2)$, and $g$ equal to $C_2$ and $-C_2$ over $G_{a_2,b_2}(i_2,j_2)_{UL}$ and
$G_{a_2,b_2}(i_2, j_2)_{LL}$, respectively. It follows that the $fg$ is equal to
$C_1C_2$ over $G_{a_1,b_2}(i,j)_{UL}$ and  $G_{a_1,b_2}(i,j)_{LR}$,
and $-C_1C_2$ over $G_{a_1,b_2}(i,j)_{UR}$ and
$G_{a_1,b_2}(i,j)_{LL}$. Thus $fg$ is an $(a_1,b_2)$-checkered function.

\begin{figure}[t] 
\centering{
\includegraphics[width=.5\linewidth]{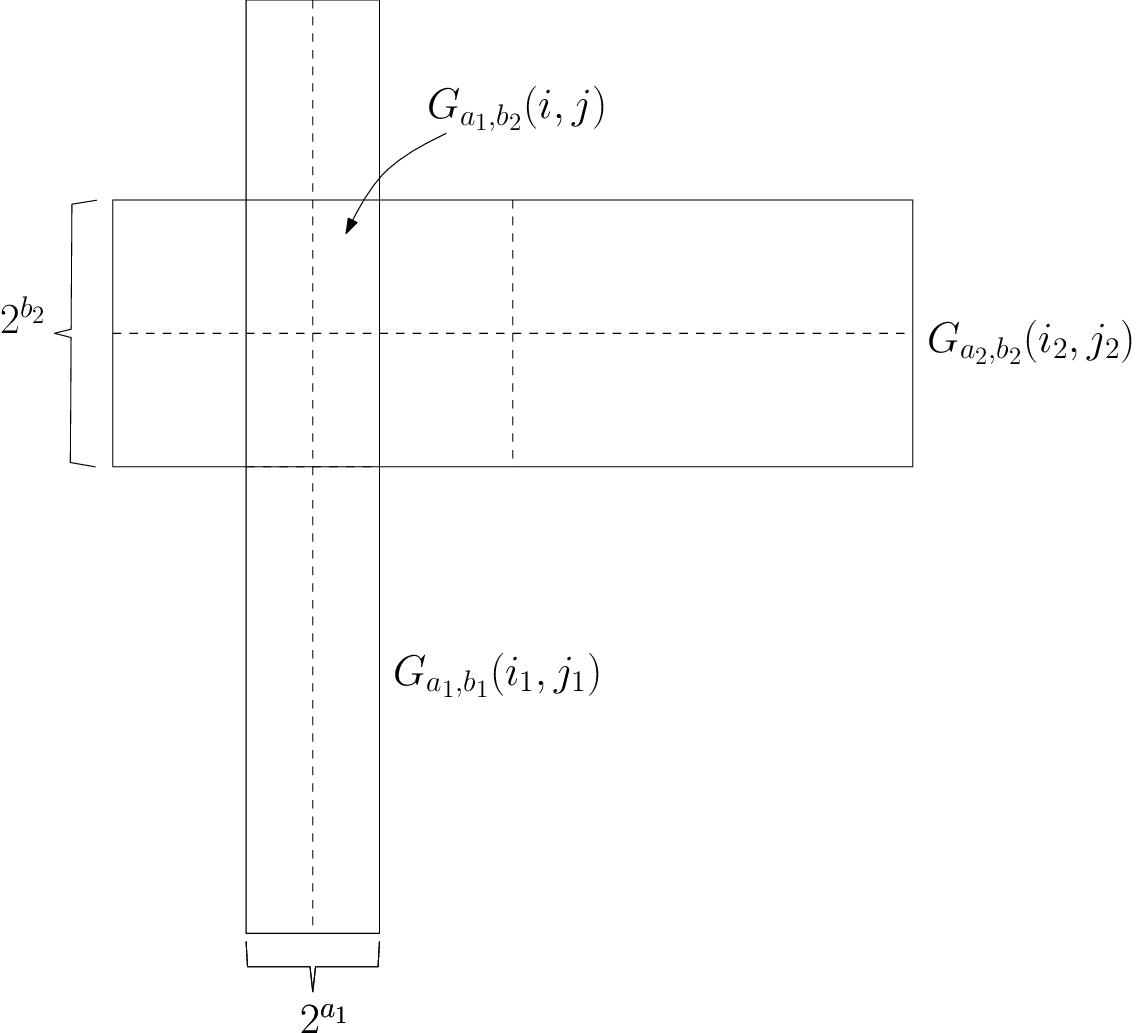}}
\caption{\label{fig:checkered} Illustration of the intersection of two cells }
\end{figure}

A direct corollary from Fact~\ref{fact:checkered} is that  the wavelet
functions are  {\em
  generalized orthogonal}:

\begin{corollary}
\label{cor:checkered}
For $0 \le k_1 < \cdots < k_l \le \log n$, the function $f_{k_1}(x) \dotsm
f_{k_l}(x)$ is a $(k_1, \log n -k_l)$-checkered. As a consequence, we have
$$\int_{[0,1)^2} f_{k_1}(x) \dotsm f_{k_l}(x) dx = 0.$$
\end{corollary}

In the remaining of the paper we assume the range of the integration is $[0,1)^2$ and the variable
of integration is $dx$ when not
specified. We define the Riesz product 
$$G(x)= -1 + \prod_{k=0}^{\log n }(\gamma f_k(x)+1),$$
where $\gamma$ is some constant to be determined later.
By the inequality 
$$\abs{\int GD} \le \int \abs{GD}  \le \sup_{x \in [0,1)^2}
\abs{D} \cdot \int \abs{G},$$
we can lower-bound the combinatorial discrepancy of $P$ as follows:
 \begin{equation} \sup_{x \in [0,1)^2 } \abs{D} \ge \left. \abs{\int GD} \middle/ \int
   \abs{G} .\right. \label{eqn:sup} \end{equation}
For the denominator $\int \abs{G}$, we have
\begin{align}
\int \abs{G} &= \int \abs{-1 + \prod_{k=0}^{\log n }(\gamma
  f_k+1)} \le1+\sum_{l=0} ^{\log n} \gamma^l \sum_{0\le
  k_1 < \ldots <k_l \le \log n} \int f_{k_1} \dotsm f_{k_l} \nonumber  \\
&=2+\sum_{l=1} ^{\log n} \gamma^l \sum_{0\le
  k_1 < \ldots <k_l \le \log n} \int f_{k_1} \dotsm f_{k_l} =2. \label{eqn:denominator}
\end{align}
The last equation is due to Corollary~\ref{cor:checkered}.  The numerator $\abs{\int G(x)D(x) dx}$ can be expressed as follow:
\begin{align}
\abs{\int GD} =& \abs{\int \left( -1 + \prod_{k=0}^{\log n
  }(\gamma f_k+1) \right) \cdot D} \nonumber \\
=& \abs{ \int  \left ( \gamma \sum _{k=0}^{\log n} f_k +\sum_{l=2}
    ^{\log n} \gamma^l \sum_{
  0 \le k_1 < \ldots <k_l \le \log n }  f_{k_1} \dotsm f_{k_l}
\right) \cdot D} \nonumber \\
\ge& \, \gamma \abs{\sum_{k=0}^{\log n} \int f_k D} -\sum_{l=2}^{\log n}
\gamma^l \abs{  \sum_{  0\le k_1 < 
      \ldots < k_l \le \log n } \int f_{k_1}  \dotsm f_{k_l} D}. \label{eqn:numerator}
\end{align}

In order to estimate $\int f_k D$, we consider the integration of a single product $ f_k(x) D(x) $ over a
$k$-canonical cell $G_{k}(i,j)$. Recall that there is exactly one
point of $P$ that lies in $G_k(i,j)$. We use $q$ to denote this point
in $P$, and $\chi(q)$ denote its color. Define horizontal
vector $u=(\frac{2^{k-1}}{n},0)$ and vertical vector $v=(0, \frac{1}{2^{k+1}})$. Then for any point $x\in
G_k(i,j)_{LL}$, points $x+u$, $x+v$ and $x+u+v$ are the analogous points
in quadrants $G_k(i,j)_{LR}$, $G_k(i,j)_{UL}$ and $G_k(i,j)_{UR}$ of
$x$, respectively (see Figure~\ref{fig:integration}). The four
analogous points defines an orthogonal rectangle.
We use $R_x$ to denote the orthogonal rectangle, and  function $R(x)$
to denote the indicator function of
point $q$ and $R_x$, that is, $R(x)=1$ if $q\in R_x$ and $R(x)=0$ if
otherwise.   We can express the integral as

\begin{align*}
\int_{G_k(i,j)}f_k(x) D(x) dx =& \int_{G_k(i,j)_{LL}}
\chi(q) \left(D(x)-D(x+u) -D(x+v)+D(x+u+v)\right) dx\\  
=& \int_{G_k(i,j)_{LL}}\chi (q) \cdot \chi(q) R(x) dx
= \int_{G_k(i,j)_{LL}} R(x) dx .
\end{align*}
The second equation is because $(D(x)-D(x+u)-D(x+v)+D(x+u+v))$ only
counts points inside $R_x$, which can only be $q$, or nothing otherwise.
Observe that $R(x)=1$ if and only if one
of  $x$'s analogous points lies inside the rectangle 
defined by $q$ and its nearest corner (see Figure~\ref{fig:integration}), so we have
\begin{align}
\int_{G_k(i,j)}f_k D =&  \int_{G_k(i,j)_{LL}}
R 
=V_P(k,i,j). \label{eqn:corner}
\end{align}
Now we can compute the first term in~\eqref{eqn:numerator}:
\begin{align}
 \gamma \abs{\sum_{k=0}^{\log n} \int f_k D} 
&= \gamma \abs{
\sum_{k=0}^{\log n}
\sum_{i=0}^{n/2^k-1} \sum_{j=0}^{2^k-1} \int_{G_k(i,j)} f_k D}
= \gamma
\abs{ \sum_{k=0}^{\log n}
\sum_{i=0}^{n/2^k-1} \sum_{j=0}^{2^k-1}  V_P(k,i,j)} \nonumber\\
&= \gamma S_P \ge c\gamma \log n. \label{eqn:first_part}
\end{align}

For the second term in~\eqref{eqn:numerator}, consider a $(k_1, \log n -k_l)$-cell
$G_{k_1, \log n-k_l}$. Note that $P$ intersects $G_{k_1, \log n
  -k_l}(i,j)$ with at most $1$ point.    By Fact~\ref{fact:checkered}, function  $f_{k_1}\dotsm f_{k_l}$ is
$(k_1, \log n-k_l)$-checkered, so following similar arguments in the
proof of equation~\eqref{eqn:corner}, we can show that the integral $\abs{\int_{G_{k_1, \log n-
    k_l}(i,j)} f_{k_1}\dotsm f_{k_l}D} $ is $0$ if $P\cap G_{k_1, \log n-k_l}
=\emptyset$ and otherwise equal to the
corner volume of $G_{k_1, k_l}(i,j)$. 
In the latter case, we can relax the corner volume to the
volume of $G_{k_1, \log n -k_l}(i,j)$, that is, $\frac{1}{2^{k_l-k_1}n}$.  Thus we
can estimate the integral as follows:
 $$\abs{\int_{G_{k_1, \log n-
    k_l}(i,j)} f_{k_1} \dotsm f_{k_l} D } \le
\frac{1}{2^{k_l-k_1}n}.$$
Since there are $n$ non-empty $(k_1, \log n -k_l )$-cells, we have
 $$\abs{\int f_{k_1}\dotsm f_{k_l} D}  \le n \cdot \frac{1}{2^{k_l-k_1}n}=\frac{1}{2^{k_l-k_1}}.$$
Now we can estimate the second term in~\eqref{eqn:numerator}:
\begin{align}
\sum_{l=2}^{\log n} \gamma^l \abs{\sum_{0\le k_1 < \ldots < k_l \le
    \log n} \int f_{k_1} \dotsm f_{k_l} D}
\le& \sum_{l=2}^{\log n} \gamma^l \sum_{0\le k_1 < \ldots < k_l \le
    \log n} \frac{1}{2^{k_l-k_1}} \nonumber\\
=& \sum_{l=2}^{\log n} \gamma^l \sum_{w=l-1}^{\log n+1}
\sum_{k_l-k_1=w}
\frac{1}{2^w}{w-1\choose l-2}. \label{eqn:second_part_1}
\end{align}
For the last equation we replace $k_l-k_1$ with a new index $w$  and
use the fact that there are ${w-1 \choose
 l-2}$ ways to choose $k_2,\ldots, k_{l-1}$ in an interval of length
$w$. Note that for a fixed $w$, there are $\log n +1 -w$ possible values
for $k_1$, so
\begin{align}
\sum_{l=2}^{\log n} \gamma^l \sum_{w=l-1}^{\log n+1}
\sum_{k_l-k_1=w}
\frac{1}{2^w}{w-1\choose l-2}
=& \sum_{l=2}^{\log n} \gamma^l \sum_{w=l-1}^{\log n+1}
\frac{\log n +1 -w}{2^w}{w-1\choose l-2}  \nonumber\\
\le& \sum_{l=2}^{\log n} \gamma^l \sum_{w=l-1}^{\log n+1}
\frac{\log n }{2^w}{w-1\choose l-2}  \nonumber\\
=& \log n \sum_{l=2}^{\log n} \gamma^l \sum_{w=l-1}^{\log n+1}
\frac{1}{2^w}{w-1\choose l-2}. \label{eqn:second_part_2}
\end{align}
By inverting the order of the summation,
\begin{align}
\log n \sum_{l=2}^{\log n} \gamma^l \sum_{w=l-1}^{\log n+1}
\sum_{k_l-k_1=w}
\frac{1}{2^w}{w-1\choose l-2}
=&\gamma^2 \log n \sum_{w=1}^{\log n +1} \frac{1}{2^{w}}
\sum_{l=2}^{w+1}{w-1 \choose l-2} \gamma^{l-2} \nonumber \\
=&\gamma^2 \log n \sum_{w=1}^{\log n +1}
\frac{1}{2^{w}}(1+\gamma)^{w-1} \nonumber \\
=& 2\gamma^2 \log n \sum_{w=1}^{\log n +1}
\left( \frac{1+\gamma}{2}\right)^{w-1} 
\le \frac{2 \gamma^2 }{1-\gamma}\log n. \label{eqn:second_part_3}
\end{align}
So
from~\eqref{eqn:first_part},~\eqref{eqn:second_part_1},~\eqref{eqn:second_part_2}
and~\eqref{eqn:second_part_3} we have 
$$\abs {\int GD} \ge c\gamma \log n-  \frac{2 \gamma^2 }{1-\gamma}\log n.$$ 
Setting $\gamma$ small enough while combining with~\eqref{eqn:sup}
and~\eqref{eqn:denominator}  completes the proof. 
\end{proof}

Now we can give a proof to Theorem~\ref{thm:discrepancy}.
By Lemma~\ref{lem:corner}, we only need to show that the corner volume
sum of any point set $P\in \Pcal_0$ is large. 
Fix $k$ and consider a $k$-canonical cell $G_k(i,j)$. Let $q$ denote the point
in $P\cap G_k(i,j)$. We define the corner $x$-distance
of  $G_k(i,j)$ to be the difference between the
$x$-coordinate of $q$ and that of  its nearest corner
of $G_k(i,j)$. The corner $y$-distance is defined in similar
manner. See Figure~\ref{fig:integration}. We use  $X(k,i,j)$ and $Y(k,i,j)$ to denote the
corner $x$-distance and corner $y$-distance, respectively.
Note that  the corner volume $V_P(k,i,j)$ is the product of $X(k,i,j)$ and
$Y(k,i,j)$.  The following fact holds for the $x$-distances of
canonical cells in a column:

\begin{fact}
\label{fact:x}
 Fix $k$ and $i$, we have $\{X(k,i,j) \mid j \in
  [2^k]\}=\{\frac{j}{n}+\frac{1}{2n}, \frac{j}{n}+\frac{1}{2n} \mid j \in  [2^{k-1}]\}$, where both are taken as
  multisets. 


\end{fact}
For a proof, note that  the $k$-canonical cell
$G_k(i,j)$ is intersecting with $2^k$ columns:
$G_0(i2^k,0), \ldots, G_0((i+1)2^k-1,0)$.
There are $2^{k}$ points in $G_k(i,0),\ldots,G_k(i, 2^k-1)$, and they
must reside in different columns. Therefore there is exactly one
point in the each of the $2^k$  columns,  and their corner
$x$-distances span from $\frac{1}{2n}$ to $\frac{2^{k-1}-1}{n}+\frac{1}{2n}$, and each value is hit
exactly twice. Similarly, we have

\begin{fact}
\label{fact:y}
Fix $k$ and $j$, we have $\{X(k,i,j) \mid i \in
  [n/2^k]\}=\{ \frac{i}{n} + \frac{1}{2n}, \frac{i}{n} + \frac{1}{2n} \mid i \in  [n/2^{k+1}]\},$ where both are
  taken as multisets.
\end{fact}
 
Now consider the product of $X(k,i,j)$ and $Y(k,i, j)$ over all
$(i,j)$ for a fixed $k$:
\begin{align*}
\prod_{i=0}^{n/2^{k}-1}\prod_{j=0}^{2^k-1}V_P(k,i,j) 
=&
\prod_{i=0}^{n/2^{k}-1}\prod_{j=0}^{2^k-1}X(k,i,j)Y(k,i,j)\\
=& \prod_{i=0}^{n/2^{k}-1}\prod_{j=0}^{2^k-1}X(k,i,j)
\cdot \prod_{j=0}^{2^k-1}\prod_{i=0}^{n/2^{k}-1}Y(k,i,j) \\
=& \prod_{i=0}^{n/2^{k}-1}\prod_{j=0}^{2^{k-1}-1}(\frac{j}{n} + \frac{1}{2n})^2
\cdot \prod_{j=0}^{2^k-1}\prod_{i=0}^{n/2^{k+1}-1}(\frac{i}{n} + \frac{1}{2n})^2.
\end{align*}
The last equation is due to Fact~\ref{fact:x} and
Fact~\ref{fact:y}. By relaxing  $ \frac{i}{n} + \frac{1}{2n} $ and $ \frac{j}{n} + \frac{1}{2n} $ to $ \frac{i+1}{2n}$ and
$\frac{j+1}{2n}$, we have
\begin{align*}
\prod_{i=0}^{n/2^{k}-1}\prod_{j=0}^{2^k-1}V_P(k,i,j) 
\ge& \prod_{i=0}^{n/2^{k}-1}\prod_{j=0}^{2^{k-1}-1}
\left(\frac{i+1}{2n} \right)^2
\cdot \prod_{j=0}^{2^k-1}\prod_{i=0}^{n/2^{k+1}-1}
\left(\frac{j+1}{2n} \right )^2\\
=\frac{1}{n^n}&\prod_{i=0}^{n/2^{k}-1}\left(\frac{(2^{k-1})!}{2^{2^{k-1}}}\right)^2
\cdot \prod_{j=0}^{2^k-1}\left(\frac{(n/2^{k+1})!}{2^{n/2^{k+1}}}\right)^2.
\end{align*}
By the inequality $x!\ge (x/e)^x$,
\begin{align*}
\prod_{i=0}^{n/2^{k}-1}\prod_{j=0}^{2^k-1}V(k,i,j) 
\ge& \frac{1}{n^{2n}} \prod_{i=0}^{n/2^{k}-1}\left( \left(
    \frac{2^{k-1}}{2e} \right)^{2^{k-1}}\right)^2
\cdot \prod_{j=0}^{2^k-1}\left( \left( \frac{n/2^{k+1}}{2e} \right)^{n/2^{k+1}} \right)^2\\
=& \frac{1}{n^{2n}} \prod_{i=0}^{n/2^k-1} \left(
    \frac{2^{k-1}}{2e} \right)^{2^{k}}
\cdot \prod_{j=0}^{2^k-1}\left(  \frac{n/2^{k+1}}{2e} \right)^{n/2^{k}}\\
=& \frac{1}{n^{2n}} \left(
    \frac{2^{k-1}}{2e} \right)^{2^{k}\cdot n/2^{k}}\cdot \left(  \frac{n/2^{k+1}}{2e} \right)^{n/2^{k}\cdot 2^k}\\
=& \frac{1}{n^{2n}} \left(
    \frac{2^k}{4e} \right)^n \cdot \left(  \frac{n/2^{k}}{4e} \right)^n
= \left(\frac{1}{16en}\right)^n.
\end{align*}
Using the inequality of geometric means, 
\begin{align*}\sum_{i=0}^{n/2^k-1}\sum_{j=0}^{2^k-1}V_P(k,i,j)  
\ge n\cdot \left(
  \prod_{i=0}^{n/2^{k}-1}\prod_{j=0}^{2^k-1}V_P(k,i,j) \right)^{1/n}
\ge n \cdot \frac{1}{16en} = \frac{1}{16e}.
\end{align*}
So the corner volume sum $S_P=\sum_{k=0}^{\log
  n}\sum_{i=0}^{n/2^k-1}\sum_{j=0}^{2^k-1} \allowbreak V(k,i,j)$ is lower bounded
by $\log n /16e$, and by Lemma~\ref{lem:corner}, Theorem~\ref{thm:discrepancy} follows.

\subsection{A bit vector representation for $\Pcal_0$}
Another nice property of $\Pcal_0$ is that we can derive the exact
number of point sets in it. The following lemma is from the
book~\cite{darnall2008results}. We sketch the proof here, as it also provides a
bit vector presentation of each binary net, which is essential in our
lower bound proof.

\begin{lemma}[\cite{darnall2008results}]
\label{lem:large_size}
The number of point sets in $\Pcal_0$ is $2^{\frac{1}{2}n\log n}$.
\end{lemma}

\begin{proof}
It is
equivalent to prove that the number of possible
ways to place $n$ points on the $n\times n$ grid such that any
$k$-canonical cell $G_k(i,j)$ has exactly $1$ point is
$2^{\frac{1}{2}n \log n}$. Our proof proceeds  by an induction on $n$. 
Let $\Pcal_0(n)$ denote the
collection of binary nets of size $n$ in a $n\times n$ grid. 

Observe that the line $y=n/2$ divides the grid $[0,1)^2$ 
into two rectangles: the upper grid $[0,1)\times [\frac{1}{2},1)$ and
the lower grid $[0,1)\times [0,\frac{1}{2})$. For $i$ even, let $R_i$
denote the rectangle defined by the union
of  $i$-th and $(i+1)$-th columns $G_0(i,0)$ and $G_0(i+1,0)$. Note that
the line $y=n/2$ divides $R_i$ into $G_1(\frac{i}{2},0)$ and
$G_1(\frac{i}{2},1)$, and therefore defines four quadrants. By the definition of
$\Pcal_0$, for any point set $P\in \Pcal_0$, the two points in
$G_0(i,0)$ and $G_1(i+1,0)$ must either reside in the lower left 
 and upper right  quadrants or in the lower right and upper left
 quadrants. There are in total $n/2$ even $i$'s, so the number of the possible
 choices is $2^{n/2}$. See Figure~\ref{fig:P0}. Note that after determining
 which half the point in each column resides in, the problem is divided into two
 sub-problems: counting the number of possible ways to place $n/2$
 points in the upper grid and the lower grid. It is easy to show that
 each sub-problem is
 identical to the problem of counting the number of point sets in  
 $\Pcal_0(n/2)$, so we have the
 following recursion:
$$\abs{\Pcal_0(n)}=2^{\frac{n}{2}}\cdot \abs{\Pcal_0(n/2)}^2.$$
Solving this recursion with $\Pcal_0(1)=1$ yields that $\abs{\Pcal_0(n)}=2^{\frac{1}{2}n\log n}$.
\end{proof}

\begin{figure*}[!t]
\centering{
\includegraphics[width=.7\linewidth]{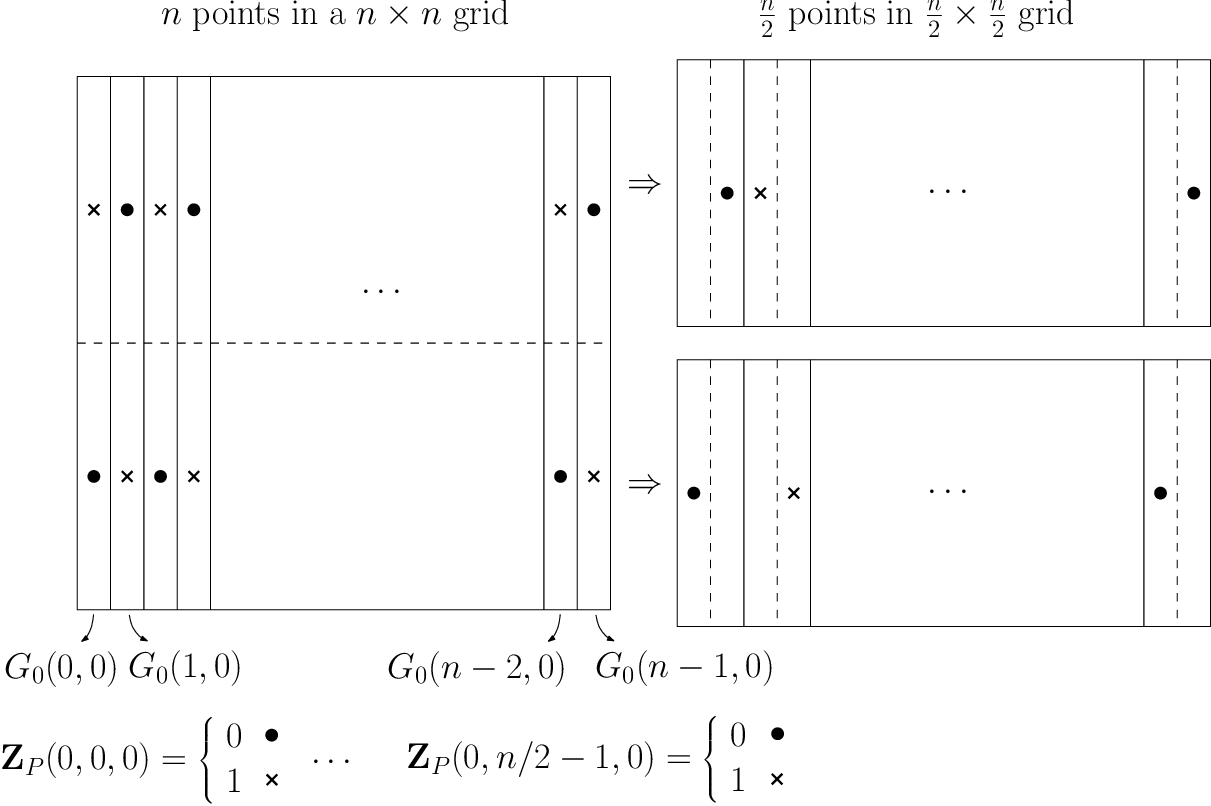}}
\caption{\label{fig:P0} Illustration of  the partition vector of $G_0$.}
\end{figure*}
 
A critical observation is that the proof of
 Lemma~\ref{lem:large_size} actually reveals
a bit vector representation for each of the point sets in
$\Pcal_0$, which will allow us to refine
 the collection
$\Pcal_0$. To see this, we define the {\em partition  vector}
$\Zbf_P$ for a point set $P\in \Pcal_0$ as follows.
 For any $(k,i,j)\in[\log n]\times [n/2^{k+1}]\times [2^k]$,
consider the $k$-canonical cells $G_k(2i,j)$ and $G_k(2i+1,j)$ and
$(k+1)$-canonical cells $G_{k+1}(i,2j)$ and
$G_{k+1}(i,2j+1)$. The two $k$-canonical cells overlap with the two
$(k+1)$-canonical cells, which defines four
quadrants. By the definition of binary nets, there are
two points in $P$  contained in these quadrants. 
We define $\Zbf_P(k,i,j)=0$ if the two points are in the lower left and
upper right quadrants and $\Zbf_P(k,i,j)=1$ if they are in the lower right
and upper left quadrants. See Figure~\ref{fig:P0}. We say the $k$-canonical cells
$G_k(2i,j)$ and $G_k(2i+1, j )$ is associated with bit
$\Zbf_P(k,i,j)$. Note that we use the triple $(k,i,j)$ as the index into
 $\Zbf_P$ 
for the ease of presentation; we can assume that the bits in $\Zbf_P$ are
 stored in for example the
lexicographic order of $(k,i,j)$.  Since the number of triples $(k,i,j)$ is 
$\frac{1}{2}n \log n$, the total number of bits in $\Zbf_P$ is
$\frac{1}{2}n \log n$. Let $\Zcal_0 = \{0,1\}^{\frac{1}{2}n \log n}$
denote  the set of all possible partition vector $\Zbf_P$'s. By the proof of
Lemma~\ref{lem:large_size},  there is a bijection between $\Zcal_0$
and $\Pcal_0$.

\subsection{Combinatorial discrepancy and corner volume distance}
Although we have proved that binary nets have large
combinatorial discrepancy, it does not yet lead us to
Theorem~\ref{thm:point_sets}.  
In this section, we will refine  $\Pcal_0$, the collection of all
binary nets, to derive a collection $\Pcal^*$,   such that
the union of any two point sets in $\Pcal^*$ has large combinatorial 
discrepancy. In order to characterize the combinatorial discrepancy of the union of
two point sets, we will need the following definition of {\em  corner
  volume distance}. 

\begin{Definition}
For two point sets $P_1, P_2 \in \Pcal_0$, the {\em corner volume distance}
of $P_1$ and $P_2$ is the summation of
$\abs{V_{P_1}(k,i,j)-V_{P_2}(k,i,j)}$, over all $(k,i,j)$. In
other words, let $\Delta(P_1, P_2)$ denote the corner volume distance of
$P_1$ and $P_2$, then
 $$\Delta(P_1, P_2)=\sum_{k=0}^{\log n}\sum_{i=0}^{n/2^k-1}\sum_{j=0}^{2^k-1}
\abs{V_{P_1}(k,i,j)-V_{P_2}(k,i,j)}. $$
\end{Definition}

The following lemma relates the combinatorial
discrepancy of the union of two point sets with their corner volume
distance:

\begin{lemma}
\label{lem:corner_volume_distance}
 Let $\Pcal^*$ be a subset of $\Pcal_0$. If there exists a constant $c$,
 such that for any two point sets $P_1, P_2 \in \Pcal_0$, 
 that their corner volume 
distance satisfies
$\Delta(P_1, P_2) \ge c \log n$, 
then $\disc(P_1\cup P_2, \Rcal_2) = \Omega(\log n)$.
\end{lemma}  

\begin{proof}
The proof follows the same framework as the proof for
Lemma~\ref{lem:corner}. 
Note that  there are exactly two points of $P_1\cup P_2$
in each $k$-canonical cell
$G_k(i,j)$, and we use $q_1, q_2$ denote the two points from $P_1$ and
$P_2$, respectively. 
We will set $f_k(x)=C$ for quadrants $G_k(i,j)_{UR}$
and $G_k(i,j)_{LL}$ and $f_k(x)=-C$ for the other two quadrants, where $C$ is determined as
follows:
$$C=\left\{
     \begin{array}{lr}
       \chi(q_1) & \textrm{if } V_{P_1}(k,i,j) \ge V_{P_2}(k,i,j); \\
 \chi(q_2) & \textrm{if } V_{P_1}(k,i,j) < V_{P_2}(k,i,j).
     \end{array}
   \right. $$
Let $D(x)$ be the combinatorial discrepancy at $x$ over $P_1\cup P_2$. By equation~\eqref{eqn:first_part}
in the proof of Lemma~\ref{lem:corner}, we
get
$$\int_{G_k(i,j)} f_kD=\left\{
   \begin{array}{lr}
        (V_{P_1}(k,i,j) +
       V_{P_2}(k,i,j)) & \textrm{if }
        \chi(q_1) = \chi(q_2); \\
\abs{V_{P_1}(k,i,j) - V_{P_2}(k,i,j)}& 
\textrm{if } \chi(q_1) \neq \chi(q_2). 
     \end{array}
   \right. $$
In either case, 
$$ \int_{G_k(i,j)} f_kD
\ge \abs{ V_{P_1}(k,i,j)-V_{P_2}(k,i,j) }.$$
And the rest of the proof follows the same argument in the proof of
Lemma~\ref{lem:corner}. 
\end{proof}

Here we briefly explain the high level idea for proving
Theorem~\ref{thm:point_sets}.  By Lemma~\ref{lem:corner_volume_distance},
it is sufficient to find a sub-collection $\Pcal^* \subseteq \Pcal_0$, such
that for any two point sets in $\Pcal^*$, their corner volume distance is
large. We will choose a subset $\Zcal_1 \subseteq \Zcal_0$, and project
each vector in $\Zcal_1$ down to a slightly shorter bit vector $\Tbf$. The
collection $\Tcal$ of all resulted bit vector $\Tbf$'s induces a
sub-collection $\Pcal_1 \subseteq \Pcal_0$, and each $\Tbf$ represents a
point set in $\Pcal_1$. Then we prove that for any two point sets $P_1, P_2
\in \Pcal_1$, there is a linear dependence between the corner volume
distance $\Delta(P_1, P_2)$ and the Hamming distance of their bit vector
representations $\Tbf_{P_1}$ and $\Tbf_{P_2}$. Finally, we show that there
is a large sub-collection of $\Tcal$ with large pair-wise Hamming
distances, and this sub-collection induces a collection of point sets
$\Pcal^*\in \Pcal_1$ in which the union of any two point sets has large
combinatorial discrepancy.

We focus on  an $(k+6, \log n-k)$-cell $G_{k+6,\log
n -k}(i,j)$, for  $k\in \{0,6, 12, \ldots, \log n-6 \}$.  Note that  $G_{k+6,\log
n -k}(i,j)$ only contains $(k+l)$-canonical cells
 for $l\in [7]$.  Let
$F_{k,i,j}(l)$ denote the set of all $(k+l)$-canonical cells in
$G_{k+6, \log n -k}(i,j)$, which can be listed as
$$F_{k,i,j}(l)=\{G_{k+l}(2^{6-l}i+s, 2^{l}j+t) \mid s\in [2^{6-l}],
t \in [2^{l}]\}.$$
Note that $\abs{F_{k,i,j}(l)}=64$ for each $l\in [7]$.
Let $Z_{k,i,j}(l)$ denote the set of indices of  bits in the partition
vector that are associated with the
some $(k+l)$-canonical cells in $G_{k+6,\log n-k}(i,j)$, for $l\in [6]$, i.e.,
$$Z_{k,i,j}(l)=\{(k+l, 2^{5-l}i+s, 2^lj+t) \mid
s\in [2^{5-l}], t \in [2^{l}]\}.$$
Define $Z_{k,i,j}$ to be the union of the $Z_{k,i,j}(l)$'s. Since
there are $32$ bits in $Z_{k,i,j}(l)$ for each $l \in [6]$, the
total number of bits in $Z_{k,i,j}$ 
is $192$ (here we use the indices in
$Z_{k,i,j}$ to denote their corresponding bits in the partition vector of
$P$, with a slightly abuse of notation).  The following fact shows the
$Z_{k,i,j}$'s partition all the
$\frac{1}{2}n\log n$ bits:

\begin{fact}
\label{fact:subset}
The number of  $Z_{k,i,j}$'s is $\frac{1}{384}n\log
n$; For different $(k,i,j)$ and $(k',i',j')$, $Z_{k,i,j}\cap Z_{k',i',j'}=\emptyset$.
\end{fact}

The proof of the above claims are fairly straightforward: 
The number of different $Z_{k,i,j}$'s is equal to the number of
different $G_{k+6,\log n -k}(i,j)$'s. For a fixed $k$, the number of
different $(k+6, \log n -k)$-cells is $n/64$, and the number of
different $k$'s is $\log n /6$, so the total number of different
$Z_{k,i,j}$'s is $\frac{1}{384}n \log n$. For the second claim, we
consider the following two cases:  If $k=k'$, we have$(i,j) \neq
(i',j')$. This implies that the two  $(k,\log n-k+6)$-cells are
disjoint, therefore the bits associated with the canonical cells
inside them are disjoint.
For $k\neq k'$,  observe that we choose $k$ and $k'$ from $\{0,6,
\ldots, \log n-6\}$, and $Z_{k,i,j}$ and $\Zbf_{k',i',j'}$ only contain bits associated with
$(k+l)$-canonical cells and $(k'+l')$-canonical cells, respectively, for $l,l'\in [6]$,
so  $Z_{k,i,j}(l)$ and $Z_{k'i'j'}(l')$ are disjoint, for
$l,l'\in [6]$.

The reason we group the bits in the partition vector  into
small subsets is that  we can view each subset 
$Z_{k,i,j}$ as a partition  vector of the cell $G_{k+6,
   \log n -k}(i,j)$,  which  allows us to manipulate the positions of the points
inside it.  More precisely, we can view
$G_{k+6,\log n -k}(i,j )$  as a $64\times 64$ grid, with 
each grid cell being a $(k,\log n -k-6)$-cell in the original $[0,1)^2$
grid. Moreover, 
a $(k+l)$-canonical cell contained in $G_{k +6 ,\log n - k}(i,j)$ corresponds
to a $l$-canonical cell in the $64\times 64$ grid. Note that
there are $64$ points in this grid, and the bits in $Z_{k,i,j}$ correspond to
the partition vector of this $64$-point set.
Now consider a $(k+3)$-canonical cell $G_{k+3}(8i,8j)$, which 
corresponds  to the lower
left $8\times 8$ grid in $G_{k+6,\log n -k}(i,j )$.  For each point set $P
\in \Pcal_0$,  there is exactly one
point in $G_{k+3}(8i,8j)$, and the bits in $Z_{k,i,j}$ encode the position
of the point on the $8\times 8$ grid. Suppose $s_1$ and $s_2$ are
two bit vectors of length $192$, such that when the bits in $Z_{k,i,j}$ are
assigned as $s_1$ (denoted $Z_{k,i,j}=s_1$), the
point in $G_{k+3}(8i,8j)$ resides in the upper left grid cell,; and when
$Z_{k,i,j} =s_2$, it resides in the grid cell
to the upper left of the center of $G_{k+3}(8i,8j)$ (see
Figure~\ref{fig:partition}). Note that by this
definition, the corner volume distance of this two point is at least
$n/8$.  Meanwhile, since there are no constraints on
the other $63$ points in $G_{k+6,\log
  n -k}(i,j )$, it is easy to show that such assignments $s_1$ and
$s_2$ indeed exist.  
  
\begin{figure*}[!t] 
\centering{
\includegraphics[width=.7\linewidth]{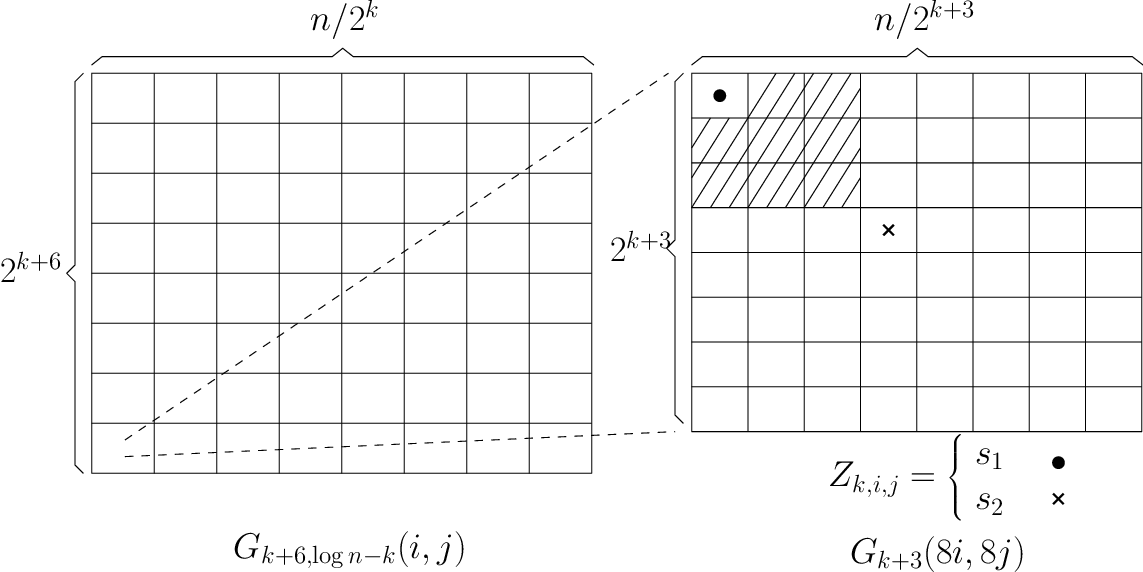}}
\caption{\label{fig:partition} Illustration of the $64\times 64$
  grid. The volume of each cell in $G_{k+3}(8i,8j)$ is $n/64$. The
  cells in shadow represent the corner volume difference of $s_1$ and $s_2$.}
\end{figure*}

By restricting the assignments of $Z_{k,i,j}$ to
$\{s_1, s_2\}$, we have created a subset $\Zcal_1$ of $\Zcal_0=\{0,1\}
^{\frac{1}{2}n \log n}$:
$$\Zcal_1=\{\Zbf \mid Z_{k,i,j}=s_1 \textrm{ or } s_2, k \in \{0, 6,
\ldots, \log n -6\}, i \in [n/2^{k+6}], j \in
  [2^k] \}.$$ 
Let $\Pcal_1$ denote the  sub-collection of
 $\Pcal_0$ that $\Zcal_1$ encode.
 By Fact~\ref{fact:subset}, the number of  $Z_{k,i,j}$'s is $\frac{1}{384} n
 \log n$, so $\abs{\Pcal_1}=2^{\frac{1}{384}n \log n}$.
Define a bit vector $\Tbf$ of length $\frac{1}{384} n \log n$, such
that $\Tbf(k,i,j)=0$ if $Z_{k,i,j}=s_1$ and $\Tbf(k,i,j)=1$ if
$\Zbf_{k,i,,j}=s_2$, then a bit vector $\Tbf$
encodes a bit vector $\Zbf \in \Zcal_1$,  and therefore encodes a
point set in $\Pcal_1$. 
Let $\Tcal=\{0,1\}^{\frac{1}{384}n \log n}$ denote the collection of all
  bit vectors $\Tbf$.  Then there is a bijection between $\Tcal$ and
  $\Pcal_1$, and 
$\abs{\Tcal}=\abs{\Pcal_1}= 2^{\frac{1}{384} n \log n}$.

Consider two point sets $P_1$ and $P_2$ in
$\Pcal_1$. Let $\Tbf_{P_1}$ and $\Tbf_{P_2}$ denote the bit vector
that encode these two point sets, respectively. The following lemma
relates the corner volume distance of $P_1$ and $P_2$ with the Hamming distance
between $\Tbf_{P_1}$ and $\Tbf_{P_2}$.

\begin{lemma}
\label{lem:hamming}
Suppose there exists a constant $c$, such that 
for any $P_1, P_2 \in \Pcal_1$, the Hamming distance
$H(\Tbf_{P_1}, \Tbf_{P_2}) \ge c n \log n$, then the corner  volume
distance between $P_1$ and $P_2$,
$\Delta(P_1, P_2)$, is  $\Omega(n^2\log n)$.
\end{lemma}
\begin{proof}
We make the following relaxation on $\Delta(P_1, P_2)$ :

\begin{align*}
\Delta(P_1, P_2) =&\sum_{k=0}^{\log n}\sum_{i=0}^{n/2^k-1}\sum_{j=0}^{2^k-1}
\abs{V_{P_1}(k,i,j)-V_{P_1}(k,i,j)} \\
\ge& \sum_{k\in\{0,6,\ldots, \log n-6\}}\sum_{i=0}^{n/2^{k+6}-1}\sum_{j=0}^{2^{k}-1}
\abs{V_{P_1}(k+3,8i,8j)-V_{P_1}(k+3,8i,8j)}.
\end{align*}

Now consider the bits  $\Tbf_{P_1}(k,i,j)$ and 
$\Tbf_{P_2}(k,i,j)$. If $\Tbf_{P_1}(k,i,j) \neq
\Tbf_{P_2}(k,i,j)$, then by the choice of $s_1$ and $s_2$ we have 
$\abs{V_{P_1}(k+3,8i,8j-V_{P_2}(k+3,8i,8j)}\ge
n/8$. So the corner volume distance $\Delta(P_1, P_2)$ is lower bounded by
the Hamming distance $H(\Tbf_{P_1}, \Tbf_{P_2})$ multiplied by  $n/8$,
and the lemma follows.
\end{proof}

The following lemma (probably folklore; we provide a proof here for completeness) states
that there is a large subset of $\Tcal$, in which the vectors are well separated in terms
of Hamming distance.

\begin{lemma}
\label{lem:large_hamming}
Let $N=\frac{1}{384}n\log n$. There is a subset $\Tcal^* \subseteq \Tcal=\{0,1\}^N$ of size
$2^{\frac{1}{16}N}$, such that for any $\Tbf_1\neq \Tbf_2 \in \Tcal^*$,
the Hamming distance $H(\Tbf_1, \Tbf_2)\ge \frac{1}{4} N$.
\end{lemma}
\begin{proof}
We embed $\Tcal$ into a graph $(V, E)$.  Each node in $V$ represents a
vector $\Tbf \in \Tcal$ , and there is edge between two nodes $\Tbf_1$
and $\Tbf_2$ if and only if $H(\Tbf_1, \Tbf_2)< \frac{1}{4}N$. By this
embedding, it is equivalent to prove that there is an independent set of
size $2^{\frac{1}{16}N}$ in $(V,E)$. 

Fix a vector $\Tbf \in \Tcal$, and consider a random vector $\Tbf'$ uniformly
drawn from $\Tcal$. It is easy to see that the Hamming distance
$H(\Tbf, \Tbf')$ follows binomial distribution. By Chernoff bound
$$\Pr[H(\Tbf,\Tbf') < \frac{1}{4}N] \le e^{-\frac{1}{16} N} \le 2^{-\frac{1}{16} N} .$$
This implies that the probability that there is an edge between $\Tbf$
and $\Tbf'$ is at most $2^{-\frac{1}{16} N}$. By the fact that $\Tbf'$
is uniformly chosen from $\Tcal$, it follows that the degree of $\Tbf$ is at
most $d=2^N\cdot 2^{-\frac{1}{16} N}=2^{\frac{15}{16}N}$. Since a graph with maximum degree $d$ must
have an independent set of size at least $\abs{V}/d$, 
there must be an independent set of size at least  $2^{\frac{1}{16}N}$.
\end{proof}

Let $\Pcal^*$ denote the collection of point sets encoded by
$\Tcal^*$. By Lemma~\ref{lem:large_hamming}, $\abs{\Pcal^*} \ge
2^{\frac{1}{16}N} = 2^{\frac{1}{6144} n\log n}$. From
Lemma~\ref{lem:corner_volume_distance} and~\ref{lem:hamming} we
know that for any two point sets $P_1\neq P_2 \in \Pcal^*$, the
combinatorial discrepancy of the union of  $P_1$ and $P_2$ is
$\Omega(\log n)$. This completes the proof of Theorem~\ref{thm:point_sets}.







\end{document}